\documentclass[10pt, letterpaper]{amsart}
\usepackage{graphicx} 
\usepackage{amsthm}
\usepackage{amssymb}
\usepackage{amsfonts}
\usepackage{amsmath}
\usepackage{mathrsfs}
\usepackage{dsfont}
\usepackage{hyperref}
\usepackage{xcolor}
\usepackage{thmtools}
\usepackage{thm-restate}
\usepackage{graphicx}
\usepackage{subcaption}

\allowdisplaybreaks[1]

\theoremstyle{plain}
\newtheorem{thm}{Theorem}[section]
\newtheorem{lemma}[thm]{Lemma}
\newtheorem{corollary}[thm]{Corollary}
\newtheorem{proposition}[thm]{Proposition}

\theoremstyle{definition}
\newtheorem{algorithm}[thm]{Algorithm}

\numberwithin{equation}{section}

\newcommand{\R}{\ensuremath{\mathds{R}}}
\newcommand{\N}{\ensuremath{\mathds{N}}}
\newcommand{\tr}{\ensuremath{\operatorname{tr}}}
\newcommand{\id}{\ensuremath{\operatorname{id}}}

\title[DBGA for entropically regularised Quantum optimal transport]{Dual Block gradient ascent for entropically regularised Quantum optimal transport}
\author{Marvin Randig}
\address[M. Randig]{Leipzig University\\Augustusplatz 10\\04109 Leipzig\\Germany}
\email{marvin.randig@mis.mpg.de}
\author{Max von Renesse}
\address[M. von Renesse]{Leipzig University\\Augustusplatz 10\\04109 Leipzig\\Germany}
\email{renesse@uni-leipzig.de}

\begin{document}

\begin{abstract}
    We present a block gradient ascent method for solving the quantum optimal transport problem with entropic regularisation similar to the algorithm proposed in \cite{FELICIANGELI2023109963} and \cite{caputo2024quantumoptimaltransportconvex}. 
    We prove a linear convergence rate based on strong concavity of the dual functional and present some results of numerical experiments of an implementation.
\end{abstract}

\maketitle
\tableofcontents

\section{Introduction}
In quantum optimal transport, non-commutative generalisations of the classical Monge-Kantorovich optimal transport problems are studied. Instead of transporting probability measures, quantum states are transported. Different approaches for generalisation are investigated, among them is a generalised formulation of the static transport problem which translates as in the commutative setting to finding an optimal coupling, cf.\ \cite{Caglioti2023} and \cite{DePalma2021}.

As for classical optimal transport, quantum optimal transport problems with entropy regularisation are studied. They have the physical interpretation as minimisation of free energy of a quantum system. 
For this problem, in \cite{FELICIANGELI2023109963} and \cite{caputo2024quantumoptimaltransportconvex}, an algorithm similar to the Sinkhorn algorithm for classical entropically regularised optimal transport based on alternating maximisation is presented. Based on these ideas, in the present paper, we propose a block gradient ascent method to solve this problem by solving its dual. 

Perhaps unsurprisingly, this method works and converges to the optimal solution at a linear rate. This is because the dual functional is strongly concave if one neglects a direction, in which it is constant. This direction however makes the result on the linear convergence rate of block coordinate descent methods from \cite[Thm.\ 3.9]{doi:10.1137/120887679} not directly applicable as it is not strongly concave in the canonical coordinates and in coordinates after eliminating the degenerate direction such that the dual functional is strongly concave, the update steps do not correspond to blocks of coordinate anymore. Still, the linear convergence rate for our method can be shown, as detailed below. 

\nocite{PEYRÉ_CHIZAT_VIALARD_SOLOMON_2019, NIPS2013_af21d0c9}

\section{Setting and presentation of the results}
\subsection{Entropy-regularised Quantum Optimal Transport}
The problem to be solved is to minimise the expected value of the energy of a composite quantum system conditioned on the states of its sub-systems. Due to possible entanglement, the state of the composite system is not fully determined by its sub-states. 

In mathematical terms, given a Hamiltonian operator $ C $ acting on a Hilbert space and two quantum states/ operators $ \rho $ and $ \sigma $, the goal is find a quantum state/ an operator $ \Gamma $ of a composite system minimising $ \tr(C\Gamma) $ subject to $ \Gamma $ being in a quantum system composed of two sub-systems in states $ \rho $ and $ \sigma $ respectively. 
For entropic regularisation, a multiple of the von Neumann entropy is added. 
Thus the objective is to find a quantum state $\Gamma$ minimising 
\begin{align}
    F(\Gamma) = \tr(C\Gamma) + \varepsilon S(\Gamma), \nonumber
\end{align}
where $S$ is the negative of the von Neumann entropy (cf.\ Subsection \ref{subsec:NeumannEntropy}) such that $ \Gamma $ is a quantum state of a system with two components being in states $\rho$ and $\sigma$ (cf.\ Subsection \ref{subsec:partialTrace}).
\subsection{Sets of operators on Hilbert spaces}
We consider two Hilbert spaces $ \mathcal{H}_1$ and $ \mathcal{H}_2$ of finite dimensions $ d_1 $ and $ d_2$. These shall be associated to sub-systems of a composite quantum system with Hilbert space $ \mathcal{H}_1 \otimes \mathcal{H}_2 $ of dimension $ d = d_1 d_2$. 

The set of self-adjoint operators on a Hilbert space $ \mathcal H$ shall be denoted $ H (\mathcal{H})$. Its subset of positive semi-definite  operators shall be denoted $ H_{\geq 0} ( \mathcal{H})$. 
The set of quantum states (or density matrices) is the set of matrices of $ H_{\geq 0}(\mathcal H)$ having trace equal to $ 1$ and is denoted $ \mathfrak{P}(\mathcal{H})$. 

\subsection{Tensor sum and partial traces}
For $U \in H(\mathcal{H}_1), V\in H(\mathcal{H}_2)$ write $ U \oplus V :=  U \otimes \id_{\mathcal{H}_2} + \id_{\mathcal{H}_1} \otimes V$. 

We also write $ H(\mathcal{H}_1) \oplus H(\mathcal{H}_2) := \left\{ A \oplus B \colon A \in H(\mathcal{H}_1), B \in H(\mathcal{H}_2) \right\}$. This shall not be understood as the direct sum in this case. Indeed, $ \left\{A \otimes \id \colon A \in H(\mathcal{H}_1) \right\} \cap \left\{\id \otimes B \colon B \in H(\mathcal{H}_2) \right\} = \{ r \id \colon r \in \R \} $. 

Given $ \Gamma \in H (\mathcal{H}_1 \otimes \mathcal{H}_2)$, one can define its first partial trace $ \tr_2(\Gamma) $ as the unique operator such that $\langle \Gamma, \ \cdot \ \otimes \id_{\mathcal{H}_2} \rangle_{\mathrm{HS}} = \langle \tr_2(\Gamma), \ \cdot \ \rangle_{\mathrm{HS}} $ as functions on $ H(\mathcal{H}_1) $ where $ \langle A, B\rangle_{\mathrm{HS}} = \tr(A^\dagger B) $ is the Hilbert-Schmidt inner product, that makes $ H(\mathcal{H}) $ an Euclidean vector space. As almost all operators present are self-adjoint, we will omit the $^\dagger $ and just write $ \tr (AB)$. Analogously, the second partial trace $ \tr_1 (\Gamma) $ is defined. If $\tr_2 (\Gamma) = \rho, \tr_1(\Gamma) = \sigma $, one writes $ \Gamma \mapsto (\rho, \sigma) $ and says $\Gamma$ has sub-states $\rho$ and $\sigma$. \label{subsec:partialTrace}

\subsection{Liftings of functions} \label{subsec:NeumannEntropy}
For every $ f: \R \to \R $ denote the lifting of $f$ to the space of self-adjoint operators $ H(\mathcal{H}) \to H(\mathcal{H}), \sum\limits_{j=1}^{\dim \mathcal{H}} \kappa_j | \xi_j \rangle \langle \xi_j | \mapsto \sum\limits_{j=1}^{\dim \mathcal H} f(\kappa_j) | \xi_j \rangle \langle \xi_j |  $ for orthonormal eigenvectors $ \{\xi_j \}_{j=1}^{\dim \mathcal H} \subset \mathcal{H} $ also by $f$. 
With this, one can define $ S: H_{\geq 0} \to \R, \Gamma \mapsto \tr(\Gamma \log(\Gamma)) $, the negative of the von Neumann entropy with the convention that $ 0 \log(0) = 0 $. 

\subsection{Primal and dual problem}

As already stated informally, given a regularisation parameter $ \varepsilon > 0$, two quantum states $ \rho \in \mathfrak{P}(\mathcal{H}_1), \sigma \in \mathfrak{P}(\mathcal{H}_2) $ and a Hamiltonian $ C \in H(\mathcal{H}_1 \otimes \mathcal{H}_2) $, the entropically regularised quantum optimal transport problem to be solved is to find $ \Gamma \in \mathfrak{P}(\mathcal{H}_1 \otimes \mathcal{H}_2) $ with $ \Gamma \mapsto (\rho, \sigma) $ minimising 
\begin{align}
    F(\Gamma) = \tr(C\Gamma) + \varepsilon S(\Gamma). \nonumber
\end{align}
Its dual problem is to find $ U \in H(\mathcal{H}_1), V\in H(\mathcal{H}_2)$ maximising 
\begin{align}
    D(U,V) &= \tr(U\rho) + \tr(V\sigma) - \varepsilon \tr \left( \exp\left( \frac{U \oplus V - C}{\varepsilon} \right) \right) + \varepsilon \nonumber \\
    &= \tr \left( (U \oplus V)(\rho \otimes \sigma) \right) - \varepsilon \tr \left( \exp\left( \frac{U \oplus V - C}{\varepsilon} \right) \right) + \varepsilon. \nonumber 
\end{align}

Abusing notation, we will also write $ D(U\oplus V) $ for $ D(U,V)$, which is well-defined by the above equation.  

We restate a part of \cite[Theorem 2.1]{FELICIANGELI2023109963}: 
\begin{thm}[Duality] \label{thm:Duality}
    There is a unique minimiser $ \Gamma \in \mathfrak{P}(\mathcal{H}_1 \otimes \mathcal{H}_2) $ of the primal problem. It holds 
    \begin{align*}
        \Gamma &= \exp \left( \frac{U \oplus V - C}{\varepsilon} \right) && \text{ on } (\ker \rho)^\perp \otimes (\ker \sigma)^\perp, \\
        \Gamma &= 0 && \text{ on } ((\ker \rho)^\perp \otimes (\ker \sigma)^\perp)^\perp 
    \end{align*}
    where $ U, V $ are maximisers of the dual problem with the restrictions of the density matrices to the complement of their kernel ($\rho |_{(\ker \rho)^\perp}$ and $ \sigma |_{(\ker \sigma)^\perp}$) as problem parameters. 
\end{thm}
Thus, in the following, we can assume that all eigenvalues of $ \rho $ and $ \sigma $ are strictly positive. This implies in particular that $ \rho \otimes \sigma $ is positive definite, i.e.\ its smallest eigenvalue is strictly positive, $ \lambda_{\min} (\rho \otimes \sigma) > 0$. 

\subsection{Gradients and marginal errors}

In order to implement the block gradient ascent method, one needs to know the gradient of the function that is to maximise. 
When $ D $ is viewed as a function on $ H(\mathcal{H}_1 \otimes \mathcal{H}_2) $, one has
    \begin{align}
        \nabla_{H(\mathcal{H}_1 \otimes \mathcal{H}_2)} D(U \oplus V) 
        &= \rho \otimes \sigma - \exp \left( \frac{U \oplus V -C}{\varepsilon} \right) ,
        \label{eq:nablaDfull}
    \end{align}
    cf. Corollary \ref{cor:nablaTrExp}.
When $D$ is viewed as a function on $H(\mathcal{H}_1) \times H(\mathcal{H}_2) $, its gradient with respect to the first variable ($U \in H(\mathcal{H}_1)$) is 
\begin{align*}
E_1(U,V):= \nabla_1 D(U,V) = \rho - \tr_2 \left( \exp \left( \frac{U \oplus V -C}{\varepsilon} \right) \right)
\end{align*}
and analogously its gradient with respect to the second variable  ($V \in H(\mathcal{H}_2)$) is 
\begin{align*}
    E_2(U,V) := \nabla_2 D (U,V) = \sigma - \tr_1 \left( \exp \left( \frac{U \oplus V -C}{\varepsilon} \right) \right).
\end{align*}
They naturally have an interpretation as marginal errors.

\subsection{Description of the algorithms} \label{subsec:descrAlg}

The overall idea is to implement a block gradient ascent method to compute sequences $ \left( U_n \right)_n \subset H(\mathcal{H}_1) $ and $ \left(V_n\right)_n \subset H(\mathcal{H}_2) $ such that $ D(U_{n+1}, V_{n+1}) \geq D( U_n, V_n) $ for all $ n$ (with strict inequality if $D(U_n,V_n)$ is not maximal). 
Making use of Theorem \ref{thm:Duality} (and some thought), this then implies 
\begin{itemize}
    \item $ U_n \oplus V_n \overset{n\to \infty}{\longrightarrow} \hat U \oplus \hat V$ for some $ \hat U \in H(\mathcal{H}_1) $ and $ \hat V \in H(\mathcal{H}_2) $,
    \item $ D(\hat U, \hat V) = \max\limits_{U \in H(\mathcal{H}_1), V \in H(\mathcal{H}_2) } D(U,V)$,
    \item for $ \hat \Gamma := \exp \left( \frac{\hat U \oplus \hat V - C}{\varepsilon} \right) = \lim\limits_{n\to \infty} \exp \left( \frac{U_n \oplus V_n - C}{\varepsilon} \right)$ it holds $ \hat \Gamma \in \mathfrak{P}\left(\mathcal{H}_1 \otimes \mathcal{H}_2\right)$ as well as $ \Gamma \mapsto (\rho, \sigma) $ and 
    \item $ F\left(\hat \Gamma \right) = \min\limits_{\Gamma \in  \mathfrak{P}(\mathcal{H}_1 \otimes \mathcal{H}_2), \Gamma \mapsto (\rho, \sigma)} F(\Gamma) $,
\end{itemize}
cf.\ the proofs of Theorem 2.3 from \cite{FELICIANGELI2023109963} and Proposition 3.12 from \cite{caputo2024quantumoptimaltransportconvex}. 

In the following description of the algorithm, $ \mathfrak{U}$; $\mathfrak{V}$; $\mathfrak{E}_1$; $\mathfrak{E}_2 $ are to be understood as variables for matrices holding values $ U_0, U_1, \ldots$; $V_0, V_1, \ldots$; $E_{1,0}, E_{1,1}, \ldots$; $E_{2,0}, E_{2,1}; \ldots $ and as only the most recently assigned value is required in each case, storing only the last value of the sequence is sufficient. For unambiguity, each computation containing these values is noted twice: Once in a way how it can implemented memory-efficiently with the variables containing only the most recent value and once as an equation of the actual values. 
\begin{algorithm}[Block Gradient Ascent Method for entropically regularised quantum optimal transport] \label{alg:bga}
    
    For a given desired precision $ \delta > 0$, one can implement the following procedure:   
    \begin{enumerate}
        \item choose $ \mathfrak{U} \equiv  U_0\in H(\mathcal{H}_1), \mathfrak{V}\equiv V_0 \in H(\mathcal{H}_2) $;
        \item compute\footnote{the function $ \nu_2$ is the inverse function of $ e^x-x-1$ as defined in Lemma \ref{lem:BoundEvs}} $ \beta \leftarrow \nu_2 \left( \frac{\tr((\rho \otimes \sigma) C)- D(\mathfrak{U}, \mathfrak{V})}{\varepsilon}  \right) \equiv \nu_2 \left( \frac{\tr((\rho \otimes \sigma) C)- D(U_0,V_0)}{\varepsilon}  \right) $;
        \item compute $ \eta_1 =  \frac{\varepsilon} {d_2} \exp \left( - \beta \right), \eta_2 =  \frac{\varepsilon}{d_1} \exp \left( - \beta \right)  $; 
        \item repeat for $n=0,1,2, \ldots $ until $ \| \mathfrak{E}_1\| \equiv \| E_{1,n} \|, \| \mathfrak{E}_2 \|  \equiv \| E_{2,n} \| < \delta $:
        \begin{enumerate}
            \item set $ \mathfrak{E}_1 \leftarrow \rho - \tr_2 \left( \exp \left( \frac{\mathfrak{U} \oplus \mathfrak{V}- C}{\varepsilon} \right) \right) \equiv E_{1,n} = \rho - \tr_2 \left( \exp \left( \frac{U_n \oplus V_n - C}{\varepsilon} \right) \right) $;
            \item update $ \mathfrak{U} \leftarrow \mathfrak{U} + \eta_1 \mathfrak{E}_1 \equiv U_{n+1} = U_n + \eta_1 E_{1,n} $;
            \item set $ \mathfrak{E}_2 \leftarrow \sigma - \tr_1 \left( \exp \left( \frac{\mathfrak{U} \oplus \mathfrak{V}- C}{\varepsilon} \right) \right) \equiv E_{2,n} = \sigma - \tr_1 \left( \exp \left( \frac{U_{n+1} \oplus V_n - C}{\varepsilon} \right) \right) $;
            \item update $ \mathfrak{V} \leftarrow \mathfrak{V} + \eta_2 \mathfrak{E}_2 \equiv V_{n+1} = V_n + \eta_2 E_{2,n} $;
        \end{enumerate}
        \item return $ \exp \left( \frac{\mathfrak{U} \oplus \mathfrak{V} - C}{\varepsilon} \right) \equiv \exp \left( \frac{U_{n+1} \oplus V_{n+1} - C}{\varepsilon} \right)$.
    \end{enumerate}
\end{algorithm}
\subsection{Convergence theorem}
\begin{restatable*}[Convergence of Algorithm \ref{alg:bga}]{thm}{convergenceAlg}\label{thm:Convergence}
    In the algorithm presented, for the sequences of matrices $  \left(U_n\right)_n $ and $ \left(V_n\right)_n $, the sequence of their tensor sums $ \left( U_n \oplus V_n \right)_n $ converges to a limit $ \hat U \oplus \hat V \in H(\mathcal{H}_1)\oplus H(\mathcal{H}_2) \subset H\left( \mathcal{H}_1 \otimes \mathcal{H}_2 \right) $. The matrices $ \hat U$ and $ \hat V $ are maximisers of the dual functional and $ \exp \left( \frac{\hat U \oplus \hat V - C}{\varepsilon} \right) $ is a minimiser of the original primal problem. 

    The algorithm converges at least linearly, i.e. 
    \begin{align*}
        \liminf\limits_{n\to +\infty} -\frac{1}{n} \log \left\| \exp \left( \frac{U_n \oplus V_n- C}{\varepsilon} \right) - \exp \left( \frac{\hat U \oplus \hat V - C}{\varepsilon} \right)  \right\|
        > 0.
    \end{align*}
\end{restatable*}

The proof of the linear convergence rate is based on the following result regarding the strong concavity of the dual functional:

\begin{restatable*}[Concavity of the dual functional]{lemma}{strongConcavity}\label{lem:DstronglyConcave}
    Let $ M \in \R $. 
    The dual functional when viewed as a function on $ H(\mathcal{H}_1) \oplus H(\mathcal{H}_2) $ is $ \gamma $-strongly concave (i.e. $-D$ is $ \gamma$-strongly convex) on $ D^{-1} ( [M,+ \infty [) $ with 
    \begin{align*}
        \gamma = \gamma (M) = \varepsilon^{-1} \exp \left( \frac{\iota_1(M)}{\varepsilon} \right),
    \end{align*}
    where $ \iota_1$ is as in Lemma \ref{lem:BoundEvs}. 
\end{restatable*}

\section{Proof of Lemma \ref{lem:DstronglyConcave} and Theorem \ref{thm:Convergence}}

In the following, we give bounds on the eigenvalues of the Hessian of $D$ to conclude that it is strongly concave and has a Lipschitz-continuous gradient and from this we deduce a linear convergence rate. 

\begin{lemma}[Spectral bounds]
    \label{lem:BoundEvs}
    There is a monotonically increasing function $ \iota_1 $ and a monotonically decreasing function $ \iota_2 $ defined on 
    the image of the dual functional 
     $\left ] - \infty, \max\limits_{U \in H(\mathcal{H}_1), V \in H(\mathcal{H}_2) } D(U,V) \right] $ 
    such that for all $ U \in H(\mathcal{H}_1), V \in H(\mathcal{H}_2) $
    \begin{align*}
        \iota_1 (D(U,V)) \leq \lambda_{\min}\left( U \oplus V - C \right) \leq \lambda_{\max} \left( U \oplus V - C \right) \leq \iota_2 (D(U,V)).
    \end{align*}
    Moreover the estimates 
    \begin{align*}
        \iota_2(x) \leq \varepsilon \nu_2 \left(\frac{\tr((\rho \otimes \sigma) C)- x}{\varepsilon} \right)
    \end{align*}
    and
    \begin{align*}
        \iota_1(x) \geq \varepsilon \nu_1 \left( x - \varepsilon -  \tr \left( C (\rho \otimes \sigma) \right) -  \varepsilon \nu_2 \left(\frac{\tr((\rho \otimes \sigma) C)- x}{\varepsilon} \right) (1 - \lambda_{\min}(\rho\otimes\sigma) ) \right)
    \end{align*}
    hold, where the function
    $ \nu_1: \left] - \infty, \varepsilon  \lambda_{\min} (\rho \otimes \sigma) \left( \log (  \lambda_{\min} (\rho \otimes \sigma)) -  \log(d) - 1 \right) \right] \to \left] - \infty, \log \left( \lambda_{\min}(\rho \otimes \sigma)\right) - \log (d )  \right] $ is the inverse function (taking the smallest values) of $ x \mapsto \varepsilon {\lambda_{\min}(\rho \otimes \sigma)} x - d\varepsilon \exp(x) $ and $ \nu_2 : \R_{\geq 0} \to \R_{\geq 0} $ is the inverse function of $ x \mapsto e^x -x-1 $, which are both monotonically increasing on their specified domains.
\end{lemma}
\begin{proof}
    Using the arguments from \cite[Prop.\ 3.3]{caputo2024quantumoptimaltransportconvex}, it is possible to show 
    \begin{align*}
        &\phantom{\geq} 
        \lambda_{\min}( U \oplus V - C)  
        \\
        &\geq \varepsilon \nu_1 \left( D(U,V) - \varepsilon
        - \tr \left( C (\rho \otimes \sigma) \right) - (1 - \lambda_{\min}(\rho\otimes\sigma) )  \lambda_{\max} (U \oplus V - C) \right) ,
    \end{align*}
    and 
    \begin{align*}
    &\phantom{\leq }\lambda_{\max}( U \oplus V - C) 
    \\
    &\leq 
    \varepsilon \nu_2 \left(\frac{\tr((\rho \otimes \sigma) C)- D(U,V)}{\varepsilon} - (d-1) \exp\left( \frac{\lambda_{\min} (U \oplus V - C)}{\varepsilon} \right)\right)
    \\
    &\leq 
    \varepsilon \nu_2 \left(\frac{\tr((\rho \otimes \sigma) C)- D(U,V)}{\varepsilon} \right). 
    \end{align*}
    The details are given in Appendix \ref{sec:Eigenvalues}. 
\end{proof}

\strongConcavity 
\begin{proof}
    By Equation (\ref{eq:nablaDfull}) and Lemma \ref{lem:hsDerivativeExp}, for all $ Z \in H( \mathcal{H}_1) \oplus H(\mathcal{H}_2) $, it holds
    
    \begin{align*}
        & \phantom{ = }
        - 
        \lim\limits_{t\to 0} \left \langle \frac{\nabla_{H(\mathcal{H}_1) \oplus H(\mathcal{H}_2)} D (U \oplus V + t Z) - \nabla_{H(\mathcal{H}_1) \oplus H(\mathcal{H}_2)} D ( U \oplus V)}{t} \middle \vert Z \right \rangle_{\mathrm{HS}}
        \\
        &= 
        -\lim\limits_{t \to 0} \left \langle \frac{\nabla_{H(\mathcal{H}_1 \otimes \mathcal{H}_2)} D (U \oplus V + t Z) - \nabla_{H(\mathcal{H}_1 \otimes \mathcal{H}_2)} D ( U \oplus V)}{t} \middle \vert Z \right \rangle_{\mathrm{HS}}
        \\
        &= 
         \left \langle \lim\limits_{t \to 0} \frac{\exp \left( \frac{U \oplus V + t Z -C}{\varepsilon} \right) - \exp \left( \frac{U \oplus V -C}{\varepsilon} \right)}{t} \middle \vert Z \right \rangle_{\mathrm{HS}}
        \\
        &=
        \frac{1}{\varepsilon}
        \left \langle \mathrm{d} \exp_{\frac{U \oplus V - C}{\varepsilon}}  \left( Z \right) \middle \vert Z \right \rangle_{\mathrm{HS}} 
        \\
        & \geq 
        \frac{1}{\varepsilon} \exp \left( \lambda_{\min} \left(\frac{U \oplus V - C}{\varepsilon} \right) \right) \| Z \|_{\mathrm{F}}^2.
    \end{align*}
    Thus each eigenvalue of the Hessian of $D$ (as a function on $ {H(\mathcal{H}_1) \oplus H(\mathcal{H}_2)} $) at $ U \oplus V $ is at least 
    \begin{align*}
    \frac{1}{\varepsilon} \exp \left( \lambda_{\min} \left(\frac{U \oplus V - C}{\varepsilon} \right) \right) \geq \varepsilon^{-1} \exp \left( \frac{\iota_1(M)}{\varepsilon} \right).
    \end{align*}
\end{proof}
On the other hand, the gradient of $ D$ is also Lipschitz-continuous on the set $D^{-1}([M, +\infty[)$ for any $M \in \R$ due to Lemma \ref{lem:BoundEvs}. This gives rise to the following result: 
\begin{corollary}[Dual functional improvement] \label{cor:dualImprovment}
    Let $ \beta :=  \iota_2 \left( \frac{\tr((\rho \otimes \sigma) C)- D(U,V)}{\varepsilon}  \right) / \varepsilon $, where $ \iota_2 $ is as defined in Lemma \ref{lem:BoundEvs}. 
    For $ 0 \leq \eta \leq  \frac{2\varepsilon}{d_2} \exp \left( - \beta \right) $ one has 
    \begin{align*}
        D(S_{\eta}(U), V)
        & \geq -\frac{1}{2\varepsilon}d_2 \exp \left(  \beta  \right) \| E_1(U) \|_{\mathrm{F}}^2 \eta^2 + \| E_1(U) \|_{\mathrm{F}}^2 \eta +  D(U,V), 
    \end{align*}
    in particular for $ \eta_0 = \frac{\varepsilon}{d_2} \exp \left( - \beta \right) $
    \begin{align*}
        D(S_\eta(U), V) 
        \geq 
         \frac{1}{2} \| E_1(U) \|_{\mathrm{F}}^2 \eta_0 + D(U,V) .
    \end{align*}
\end{corollary}
As the function $D$ is not strongly concave in ${H(\mathcal{H}_1) \times H(\mathcal{H}_2)}$-coordinates as there is the degenerate direction $ (\id, -\id) \in {H(\mathcal{H}_1) \times H(\mathcal{H}_2)}$, one cannot directly use the result on the convergence rate of block coordinate descent methods from \cite{doi:10.1137/120887679}. We use the following estimate: 

\begin{proposition}[Relation between marginal error and dual functional] \label{prop:MarginalDual}
    Let $ M \in \R $. 
    For $ U_0 \in H(\mathcal{H}_1), V_0 \in H(\mathcal{H}_2) $, one has
    \begin{align*}
        \max\limits_{U \in H(\mathcal{H}_1), V \in H(\mathcal{H}_2) } D(U,V) - D\left( U_0, V_0 \right) \leq 
        \frac{\left\| \nabla_1 D(U_0, V_0) \right\|_{\mathrm{F}}^2 + \left\| \nabla_2 D(U_0, V_0) \right\|_{\mathrm{F}}^2}{2\gamma}
    \end{align*}
    with $ \gamma = \gamma(D(U_0, V_0)) $ as in Lemma \ref{lem:DstronglyConcave}. 
\end{proposition}
\begin{proof}
    
    As $ D $ when viewed as a function on $ H(\mathcal{H}_1) \oplus H(\mathcal{H}_2)$ is ${\gamma(D(U_0 \oplus V_0))}$-strongly concave on $ D^{-1} ([D(U_0 \oplus V_0), +\infty[) $ and the set $ D^{-1} \left([D(U_0 \oplus V_0), +\infty[\right) $ is convex by Lemma \ref{lem:DstronglyConcave}, one gets 
    \begin{align*}
        \max\limits_{U \in H(\mathcal{H}_1), V \in H(\mathcal{H}_2) } D(U \oplus V) - D\left( U_0 \oplus V_0 \right) \leq  \frac{\left\| \nabla_{H(\mathcal{H}_1) \oplus H(\mathcal{H}_2)} D(U_0 \oplus V_0) \right\|_{\mathrm{F}}^2}{2\gamma(D(U_0\oplus V_0))}.
    \end{align*} 

    One verifies 
    \begin{align*}
        \nabla_{H(\mathcal{H}_1 \oplus \mathcal{H}_2)} D(U \oplus V) = \frac{\nabla_1 D(U,V)}{d_2} \oplus \left( \frac{\nabla_2 D(U,V)}{d_1} - \frac{\tr\left( \nabla_1 D(U,V) \right) \id }{d} \right)
    \end{align*}
    as computation shows
    \begin{align*}
        & \phantom{=}
        \left \langle 
        \frac{\nabla_1 D(U,V)}{d_2} \oplus \left( \frac{\nabla_2 D(U,V)}{d_1} - \frac{\tr\left( \nabla_1 D(U,V) \right) \id }{d} \right)
        \middle \vert 
        A \oplus B 
        \right \rangle_{\mathrm{HS}} 
        \\
        &=
        \left \langle 
        \nabla_1 D(U,V) 
        \middle \vert 
        A 
        \right \rangle_{\mathrm{HS}}
        + 
        \left \langle 
        \nabla_2 D(U,V)
        \middle \vert 
        B 
        \right \rangle_{\mathrm{HS}}
        \\
        &=
        \left\langle 
        \nabla_{H(\mathcal{H}_1 \otimes \mathcal{H}_2)} D(U \oplus V) 
        \middle \vert
        A \oplus B
        \right \rangle_{\mathrm{HS}}
    \end{align*}
    for all $ A\in H( \mathcal{H}_1), B \in H(\mathcal{H}_2) $. 
    One can then compute 
    \begin{align*}
        &\phantom{=}
        \left\| \nabla_{H(\mathcal{H}_1) \oplus H(\mathcal{H}_2)} D(U_0 \oplus V_0) \right\|_{\mathrm{F}}^2
        \\
        &=
        \frac{\|\nabla_1 D(U,V) \|_{\mathrm{F}}^2 }{d_2} + \frac{\|\nabla_2 D(U,V) \|_{\mathrm{F}}^2 }{d_1} 
        - \frac{\tr\left( \nabla_1 D(U,V) \right)^2}{d}
        \\
        & \leq 
        \|\nabla_1 D(U,V) \|_{\mathrm{F}}^2 + \|\nabla_2 D(U,V) \|_{\mathrm{F}}^2
    \end{align*}
    and the claim follows.
    
\end{proof}

\convergenceAlg

\begin{proof}
    The proof of convergence is as outlined in Subsection \ref{subsec:descrAlg} and completely analogous to the proofs of Theorem 2.3 from \cite{FELICIANGELI2023109963} and Proposition 3.12 from \cite{caputo2024quantumoptimaltransportconvex}. 
    We now show the linear convergence rate. 
    
    By Corollary \ref{cor:dualImprovment} (and its analogue for the second variable), one knows that for all $ n \in \N $
    \begin{align}
        D(U_n, V_n) \leq D(U_{n+1}, V_n) \leq  D(U_{n+1}, V_{n+1}). \label{eq:DualMonotone}
    \end{align}
    By the Lipschitz-continuity of the gradient of $D$, one has 
    \begin{align*}
        C_1 \|E_1 ( U_{n+1}, V_n ) \| \leq \| E_1 (U_n, V_n) \|
    \end{align*}
    for some $ C_1 > 0$. 
    Also by Proposition \ref{prop:MarginalDual}, 
    \begin{align*}
        \|E_1 ( U_{n+1}, V_n ) \|^2 + \| E_2 ( U_{n+1}, V_n ) \|^2 \geq C_2 (D(\hat U, \hat V) - D(U_{n+1}, V_n))
    \end{align*}
    for some $ C_2 > 0$. Thus 
    \begin{align*}
        \|E_1 ( U_n, V_n ) \|^2 + \| E_2 ( U_{n+1}, V_n ) \|^2 
        &\geq 
        C_1\|E_1 ( U_n, V_{n+1} ) \|^2 + \| E_2 ( U_{n+1}, V_n ) \|^2 
        \\
        &\geq C_3 (D(\hat U, \hat V) - D(U_{n+1}, V_n))
    \end{align*}
    for some $ C_3 > 0$ (e.g.\ $C_3 = \min\{1, C_1\} C_2$), but since 
    \begin{align*}
        D(U_{n+1}, V_{n+1}) 
        &\geq D(U_{n+1}, V_n)+C_4 \|E_2(U_{n+1}, V_n) \|^2
        \\
        &\geq D(U_n, V_n)+C_4 \|E_2(U_{n+1}, V_n)\|^2+C_5 \|E_1(U_n, V_n)\|^2
    \end{align*}
    for some $C_4, C_5 > 0$
    by Corollary \ref{cor:dualImprovment}, one gets 
    \begin{align}
        D(\hat U, \hat V) - D(U_{n+1}, V_{n+1})
        \leq
        D(\hat U, \hat V) - D(U_{n+1}, V_n) \label{eq:dualDiffMonotone}
    \end{align}
    by Equation (\ref{eq:DualMonotone}) and 
    \begin{align}
        D(\hat U, \hat V) - D(U_{n+1}, V_{n+1})
        \leq
        D(\hat U, \hat V) - D(U_n, V_n) - C_6 (D(\hat U, \hat V) - D(U_{n+1}, V_n)) \label{eq:daulDiffEstimate}
    \end{align} 
    for some $ C_6 > 0$
    because 
    \begin{align*}
        D(U_{n+1}, V_{n+1}) - D (U_n, V_n) 
        &\geq C_7 (\|E_1 ( U_n, V_n ) \|^2 + \| E_2 ( U_{n+1}, V_n ) \|^2) 
        \\
        &\geq C_6 (D(\hat U, \hat V) - D(U_{n+1}, V_n))
    \end{align*}
    for $ C_7 = \min \{C_4, C_5\} > 0$ and $ C_6 = C_7 C_3 $. 
    Then Equation (\ref{eq:dualDiffMonotone}) and Equation (\ref{eq:daulDiffEstimate}) yield the linear convergence: As $$ D(\hat U, \hat V) - D(U_{n+1}, V_n) \in [D(\hat U, \hat V) - D(U_{n+1}, V_{n+1}),D(\hat U, \hat V) - D(U_n, V_n)] $$ by Equation (\ref{eq:DualMonotone}),
    \begin{align*}
    &\phantom{\leq}
        D(\hat U, \hat V) - D(U_{n+1}, V_{n+1})
        \\
        &\leq
        D(\hat U, \hat V) +
        \min \left\{
         - D(U_{n+1}, V_n), 
         - D(U_n, V_n) - C_6 (D(\hat U, \hat V) - D(U_{n+1}, V_n))
        \right\}
        \\
        & \leq 
        \max\limits_{x \in [D(\hat U, \hat V) - D(U_{n+1}, V_{n+1}),D(\hat U, \hat V) - D(U_n, V_n)]} \min \left\{
        x, 
        D(\hat U, \hat V) - D(U_n, V_n) - C_6 x
        \right\}
        \\& =
        \frac{1}{1 + C_6} \left( D(\hat U, \hat V) - D(U_n, V_n) \right) 
    \end{align*}
    and as
    \begin{align*}
        \frac{\gamma(D(U_0, V_0))}{2} \| (U_n \oplus V_n) - (\hat U \oplus \hat V) \|_{\mathrm{F}}^2 \leq D\left( \hat U, \hat V\right) - D(U_n, V_n) 
    \end{align*}
    by the strong concavity (Lemma \ref{lem:DstronglyConcave}),
     $ U_n \oplus V_n $ converges also linearly to $ \hat U \oplus \hat V$. Then the function $ U \oplus V \mapsto \exp\left( \frac{U\oplus V - C}{\varepsilon} \right) $is Lipschitz-continuous on the compact set $ D^{-1}([D(U_0,V_0), + \infty]) $ so that likewise $ \exp \left( \frac{U_n \oplus V_n - C}{\varepsilon} \right) $ is linearly convergent to $\exp \left( \frac{\hat U \oplus \hat V - C}{\varepsilon} \right) $. 
\end{proof}

\section{Numerical Experiments}

\begin{figure}[htbp]
    \centering
    \begin{minipage}{0.45\textwidth}
        \centering
        \includegraphics[width=\textwidth]{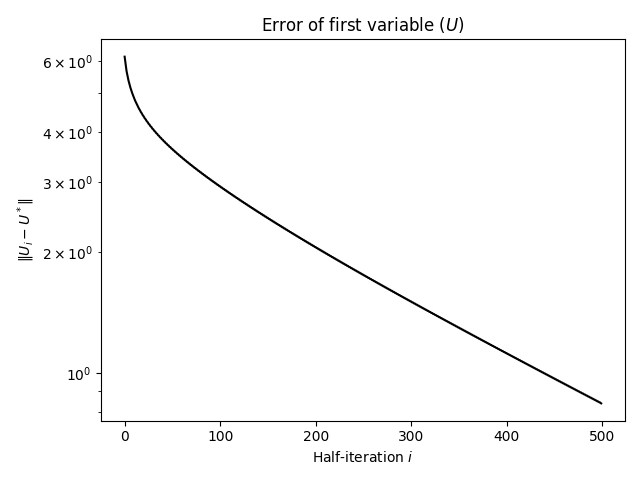}
    \end{minipage}\hfill
    \begin{minipage}{0.45\textwidth}
        \centering
        \includegraphics[width=\textwidth]{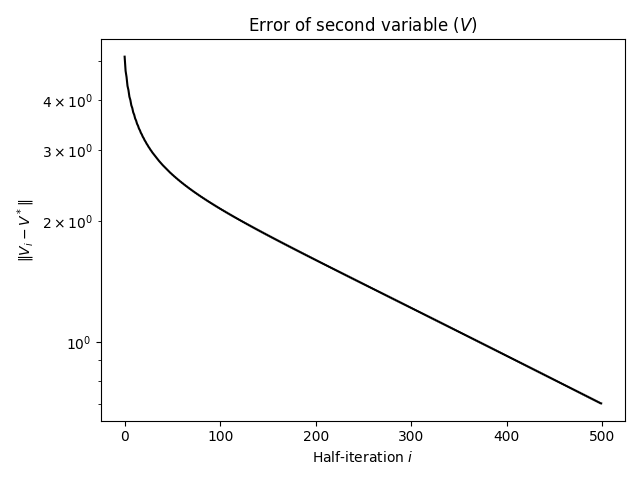}
    \end{minipage}
    \vskip\baselineskip  
    \begin{minipage}{0.45\textwidth}
        \centering
        \includegraphics[width=\textwidth]{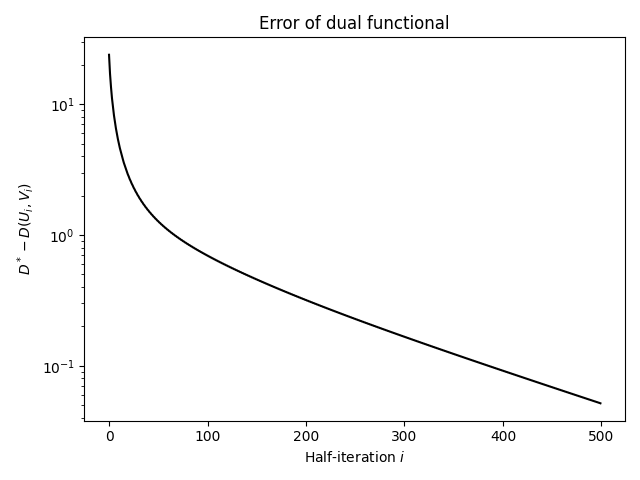}
    \end{minipage}\hfill
    \begin{minipage}{0.45\textwidth}
        \centering
        \includegraphics[width=\textwidth]{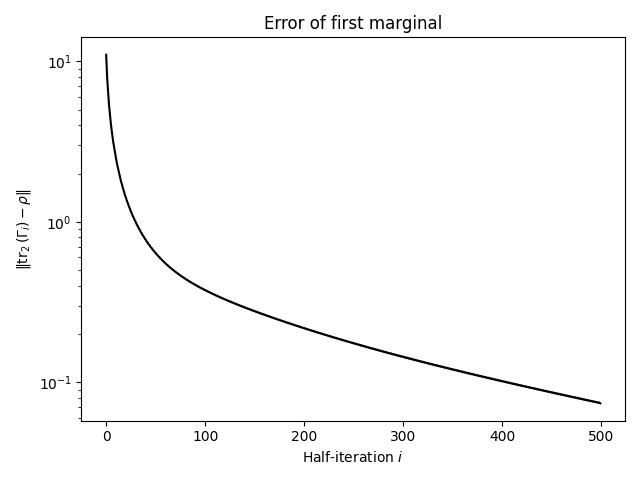}
    \end{minipage}
    \vskip\baselineskip  
    \begin{minipage}{0.45\textwidth}
        \centering
        \includegraphics[width=\textwidth]{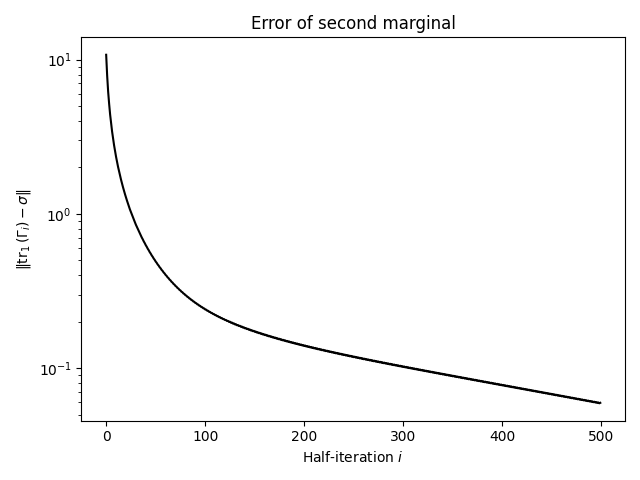}
    \end{minipage}\hfill
    \begin{minipage}{0.45\textwidth}
        \centering
        \includegraphics[width=\textwidth]{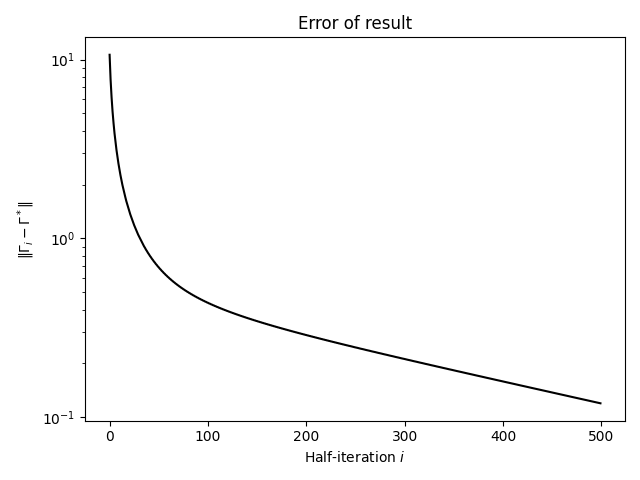}
    \end{minipage}
    \caption{These plots show the evolution of the errors of the first and second variable, dual functional, first and second marginal and result for a run of the first version of the algorithm for two $2 \times 2$ matrices as marginals and a $4\times 4$ matrix as Hamiltonian. In each of the plots, the corresponding value is plotted whenever $U$ or $V$ changes. The first 500 out of 6169 values are shown.}
    \label{fig:run1}
\end{figure}

In this section, we show some results of a run of Algorithm \ref{alg:bga}.  In Figure \ref{fig:run1},  errors are plotted, confirming empirically that after an initial phase the  convergence rate is linear.  For comparison with the optimal value, the output of the algorithm was used.  
The problem parameters were 
\begin{center}
{\small 
\begin{gather*}
    \varepsilon = 2.1440887263813604,
    \\
    \rho = \begin{pmatrix}
        0.56136249 & 0.27057949-0.07380046i \\
        0.27057949+0.07380046i & 0.43863751
    \end{pmatrix},
    \\
~\\
    \sigma = \begin{pmatrix}
        0.52286971 & 0.15823801+0.10230994i \\
        0.15823801-0.10230994 & 0.47713029
    \end{pmatrix}
    \end{gather*}
    }
{\tiny 
    \begin{gather*}
C =  \begin{pmatrix}
        -2.86929025 &  1.47503849-0.49036513i & 0.77301277-1.97328837i & -0.23265498+2.27434098i \\
        1.47503849+0.49036513i & 0.59776252 & 1.23103852-0.20920128i & 1.33284629-0.70580861i \\
        0.77301277+1.97328837i & 1.23103852+0.20920128i & 0.68185949 & -0.34221   -2.08459979i \\
        -0.23265498-2.27434098i & 1.33284629+0.70580861i & -0.34221   +2.08459979 &  1.62863414
    \end{pmatrix}.
\end{gather*}
}
\end{center}
\smallskip 
With $U_0 = V_0 = 0$ and stopping parameter $\delta = 10^{-8}$, the algorithm terminated after 3084 iterations with the result 
\smallskip 
{
\tiny
\begin{align*}
\begin{pmatrix}
    0.35691051 & 0.05291375+0.0295855421i & 0.14517306-0.00787450092i & 0.0973649 -0.105676301i \\
    0.05291375-0.0295855421i & 0.20445197 & -0.00712014-0.0616189999i & 0.12540642-0.0659259594i \\
    0.14517306+0.00787450092i & -0.00712014+0.0616189999i & 0.1659592 & 0.10532426+0.0727243916i \\
    0.0973649 +0.105676301i & 0.12540642+0.0659259594i & 0.10532426-0.0727243916i & 0.27267832
\end{pmatrix}.
\end{align*}
}

\section{Acknowledgements}
MVR thanks Bernd Sturmfels, Simon Telen, and Fran\c{c}ois-Xavier Vialard for stimulating discussions prior to this study. This work is partially  supported by BMBF (Federal Ministry of Education and Research) in DAAD project 57616814 (SECAI, School of Embedded Composite AI).

\nocite{MR4578664}

\bibliographystyle{amsalpha}
\bibliography{References}

\clearpage
\appendix

\section{Derivatives of the matrix exponential}

Fundamental for the upcoming investigations of this section is the following result due to \cite{PhysRev.73.1020}:
\begin{lemma}[Derivative of matrix exponential] \label{lem:derivativeMatrixExp}
    One has 
    \begin{align*}
    \mathrm{d} \exp _A (B) = 
        \lim\limits_{t\to 0} \frac{\exp(A+tB) - \exp(A)}{t} = \int\limits_0^1 \exp((1-s)A)B\exp(sA) \mathrm{d} s.
    \end{align*}
\end{lemma}
This enables a proof of Equation (A.5) from \cite{caputo2024quantumoptimaltransportconvex} in the case of entropic regularisation: 
\begin{corollary}[Gradient of the trace of the exponential] \label{cor:nablaTrExp}
    For every $ A \in H(\mathcal H) $, one has 
    \begin{align*}
        \nabla \tr \exp (A) = \exp (A). 
    \end{align*}
\end{corollary}
\begin{proof}
See \cite[Sec. 2.2]{MR2681769} or use Lemma \ref{lem:derivativeMatrixExp} and the cyclicity of the trace to show that $ \langle \nabla \tr \exp (A) | B \rangle_{\mathrm{HS}}  = \langle \exp(A) | B \rangle_{\mathrm{HS}} $ for every $ B \in H(\mathcal{H}) $. 
\end{proof}

\begin{lemma}[Estimate on the differential of the matrix exponential] \label{lem:hsDerivativeExp}
    For $ A, B \in H(\mathcal{H}) $, one has 
    \begin{align*}
        \exp(\lambda_{\min}(A)) \|B\|_{\mathrm{F}}^2
        \leq
        {
        \left\langle \mathrm{d} \exp_A (B) \middle| B \right\rangle_{\mathrm{HS}} 
        }{} 
        \leq \exp(\lambda_{\max}(A)) \|B\|_{\mathrm{F}}^2. 
    \end{align*}
\end{lemma}
\begin{proof}
    Let $ \{ \xi_j \}_{j=1}^d$ be orthonormal eigenvectors of $ A$ such that $ A = \sum\limits_{j=1}^d \kappa_j |\xi_j \rangle \langle \xi_j | $ with eigenvalues $ \kappa_1, \ldots , \kappa_d \in \R $. 
    Then for every $ j$
    \begin{align*}
        &\phantom{\in}
        \left\langle \xi_j \middle | \exp((1-s)A)B\exp(sA)   B  \middle | \xi_j \right \rangle_{\mathcal H} 
        = 
        \exp((1-s) \kappa_j )
        \left\|\exp((s/2)A) B \xi_j \right\|_{\mathcal H}^2 
        \\
        & \in \left[ \exp\left((1-s) \lambda_{\min}(A) \right)
        \left\|\exp((s/2)A) B \xi_j \right\|_{\mathcal H}^2 ,  
        \right. \\&\phantom{\in [} \left.
        \exp\left((1-s) \lambda_{\max}(A) \right)
        \left\|\exp((s/2)A) B \xi_j \right\|_{\mathcal H}^2 \right]
    \end{align*}
    and as $ \left\|\exp((s/2)A) B \xi_j \right\|_{\mathcal H}^2 \in \left[\exp \left(s \lambda_{\min} (A) \right) \left\| B \xi_j \right\|_{\mathcal H}^2, \exp \left(s \lambda_{\max} (A) \right) \left\| B \xi_j \right\|_{\mathcal H}^2  \right] $, one obtains
    \begin{align*}
        & \phantom{\in }
        \left\langle \xi_j \middle | \exp((1-s)A)B\exp(sA)   B  \middle | \xi_j \right \rangle_{\mathcal H} 
        \\
        &\in \left[\exp \left( \lambda_{\min} (A) \right) \left\| B \xi_j \right\|_{\mathcal H}^2, \exp \left( \lambda_{\max} (A) \right) \left\| B \xi_j \right\|_{\mathcal H}^2  \right]
    \end{align*}
    and hence
    \begin{align*}
        & \phantom{=}
        \tr \left( \exp((1-s)A)B\exp(sA)   B \right)
        \\
        &= \sum\limits_{j=1}^d \left\langle \xi_j \middle | \exp((1-s)A)B\exp(sA)   B  \middle | \xi_j \right \rangle_{\mathcal H} 
        \\
        & \in \left[ \exp \left(\lambda_{\min} (A) \right) \sum\limits_{j=1}^d \left\| B \xi_j \right\|_{\mathcal H}^2 , \exp \left(\lambda_{\max} (A) \right) \sum\limits_{j=1}^d \left\| B \xi_j \right\|_{\mathcal H}^2 \right]
        \\
        & = \left[ \exp \left(\lambda_{\min} (A) \right)  \left\| B  \right\|_{\mathrm F}^2 , \exp \left(\lambda_{\max} (A) \right) \left\| B  \right\|_{\mathrm F}^2 \right]
    \end{align*}
    because $  \left\| B  \right\|_{\mathrm F}^2 = \langle B | B \rangle_{\mathrm{HS}} = \tr \left( B^2 \right) = \sum\limits_{j=1}^d \left\langle \xi_j \middle | B^2 \middle | \xi_j \right\rangle_{\mathcal H} = \sum\limits_{j=1}^d \left\langle B \xi_j  \middle | B \xi_j \right\rangle_{\mathcal H} = \sum\limits_{j=1}^d \left\| B \xi_j \right\|_{\mathcal H}^2  $.
    From this, one sees using Lemma \ref{lem:derivativeMatrixExp}
    \begin{align*}
        \left\langle \mathrm{d} \exp_A (B) \middle| B \right\rangle_{\mathrm{HS}} 
        &=  \int\limits_0^1 \tr \left( \exp((1-s)A)B\exp(sA)   B  \right) \mathrm{d} s
        \\
        & \in \left[ \exp \left(\lambda_{\min} (A) \right)  \left\| B  \right\|_{\mathrm F}^2 , \exp \left(\lambda_{\max} (A) \right) \left\| B  \right\|_{\mathrm F}^2 \right].
    \end{align*}
\end{proof}

\section{Eigenvalue estimates} \label{sec:Eigenvalues}
\begin{lemma}[Lower bound on eigenvalues] \label{lem:evLowerBound}
    For all $ U \in H(\mathcal{H}_1)$ and $ V \in H(\mathcal{H}_2)$ 
    \begin{align*}
        &\phantom{\geq} 
        \lambda_{\min}( U \oplus V - C)  
        \\
        &\geq \varepsilon \nu_1 \left( D(U,V) - \varepsilon
        - \tr \left( C (\rho \otimes \sigma) \right) - (1 - \lambda_{\min}(\rho\otimes\sigma) )  \lambda_{\max} (U \oplus V - C) \right) ,
    \end{align*}
    holds, where the function
    $ \nu_1: \left] - \infty, \varepsilon  \lambda_{\min} (\rho \otimes \sigma) \left( \log (  \lambda_{\min} (\rho \otimes \sigma)) -  \log(d) - 1 \right) \right] \to \left] - \infty, \log \left( \lambda_{\min}(\rho \otimes \sigma)\right) - \log (d )  \right] $ is the inverse function (taking the smallest values) of $ x \mapsto \varepsilon {\lambda_{\min}(\rho \otimes \sigma)} x - d\varepsilon \exp(x) $,
    which is monotonically increasing on the specified domain. 
\end{lemma}
\begin{proof}
    One can compute in the eigenbasis of $ U \oplus V - C = \sum\limits_{j=1}^d \kappa_j |\xi_j \rangle \langle \xi_j | $ with $ \kappa_1 \leq \kappa_2 \leq \ldots \leq \kappa_d $: 
    \begin{align*}
        D(U,V)
        &= \tr \left( ( U \oplus V) (\rho \otimes \sigma) \right) - \varepsilon \tr \left( \exp\left( \frac{U \oplus V - C}{\varepsilon} \right) \right) + \varepsilon
        \\
        &= \tr \left( (U \oplus V - C) (\rho \otimes \sigma) \right) - \varepsilon \tr \left( \exp\left( \frac{U \oplus V - C}{\varepsilon} \right) \right) + \varepsilon
        \\
        &\phantom{=} + \tr \left( C (\rho \otimes \sigma) \right)
        \\
        &=
        \sum\limits_{j=1}^d \kappa_j \langle \xi_j | (\rho \otimes \sigma) | \xi_j \rangle - \varepsilon \tr \left( \exp\left( \frac{U \oplus V - C}{\varepsilon} \right) \right)  + \varepsilon
        \\
        &\phantom{=} + \tr \left( C (\rho \otimes \sigma) \right).
    \end{align*}
    Now 
    $
        \exp\left( \frac{U \oplus V - C}{\varepsilon} \right)
        =
        \exp\left( \sum\limits_{j=1}^d \frac{\kappa_j}{\varepsilon} |\xi_j \rangle \langle \xi_j | \right) = \sum\limits_{j=1}^d \exp\left( \frac{\kappa_j}{\varepsilon} \right) |\xi_j \rangle \langle \xi_j |
    $
    and hence 
    \begin{align*}
        D(U,V) &= \sum\limits_{j=1}^d \kappa_j \langle \xi_j | (\rho \otimes \sigma) | \xi_j \rangle - \varepsilon \exp\left( \frac{\kappa_j}{\varepsilon} \right) + \varepsilon 
        + \tr \left( C (\rho \otimes \sigma) \right)
        \\
        & = \kappa_1 \langle  \xi_1 | (\rho \otimes \sigma) | \xi_1 \rangle 
        + \sum\limits_{j=2}^d \kappa_j \langle \xi_j | (\rho \otimes \sigma) | \xi_j \rangle 
        \\ 
        &\phantom{=}
        - \sum\limits_{j=1}^d \varepsilon  \exp\left( \frac{\kappa_j}{\varepsilon} \right)  + \varepsilon 
        + \tr \left( C (\rho \otimes \sigma) \right)
        \\
         & \leq \kappa_1  \langle \xi_1 | (\rho \otimes \sigma) | \xi_1 \rangle 
        + \sum\limits_{j=2}^d \kappa_d \langle \xi_j | (\rho \otimes \sigma) | \xi_j \rangle 
        \\ 
        &\phantom{=}
        - \sum\limits_{j=1}^d \varepsilon  \exp\left( \frac{\kappa_j}{\varepsilon} \right)  + \varepsilon
        + \tr \left( C (\rho \otimes \sigma) \right)
    \end{align*}
    because $ \rho \otimes \sigma $ is positive definite. 
    As the exponential map is monotone and  $ \sum\limits_{j=1}^d  \left\langle \xi_j | \rho \otimes \sigma | \xi_j \right\rangle = \tr \left( \rho \otimes \sigma \right) = 1 $, we get 
    \begin{align*}
        D(U,V)
        & \leq \kappa_1 \lambda_{\min}(\rho \otimes \sigma) + \kappa_d (1 - \lambda_{\min}(\rho\otimes\sigma) ) - d \varepsilon \exp\left(\frac{\kappa_1}{\varepsilon}\right) 
        \\
        & \phantom{\leq} + \varepsilon
        + \tr \left( C (\rho \otimes \sigma) \right),
    \end{align*}
    so 
    \begin{align*}
        &\phantom{\leq }D(U,V) - \varepsilon
        - \tr \left( C (\rho \otimes \sigma) \right) - (1 - \lambda_{\min}(\rho\otimes\sigma) )  \lambda_{\max} (U \oplus V - C)  
        \\
        &\leq 
        \varepsilon \lambda_{\min}(\rho \otimes \sigma) \frac{\lambda_{\min}(U \oplus V - C)}{\varepsilon} - d \varepsilon \exp \left( \frac{\lambda_{\min} (U \oplus V - C)}{\varepsilon} \right),
    \end{align*}
    which yields the desired estimate. 
\end{proof}
\begin{lemma}[Upper bound on eigenvalues] \label{lem:evUpperBound}
    One has 
    \begin{align*}
    &\phantom{\leq }\lambda_{\max}( U \oplus V - C) 
    \\
    &\leq 
    \varepsilon \nu_2 \left(\frac{\tr((\rho \otimes \sigma) C)- D(U,V)}{\varepsilon} - (d-1) \exp\left( \frac{\lambda_{\min} (U \oplus V - C)}{\varepsilon} \right)\right)
    \\
    &\leq 
    \varepsilon \nu_2 \left(\frac{\tr((\rho \otimes \sigma) C)- D(U,V)}{\varepsilon} \right)
    \end{align*}
    where $ \nu_2 : \R_{\geq 0} \to \R_{\geq 0} $ is the inverse function of $ x \mapsto e^x -x-1 $. 
\end{lemma}
\begin{proof}
    One can again compute in the eigenbasis of $ U \oplus V - C = \sum\limits_{j=1}^d \kappa_j |\xi_j \rangle \langle \xi_j | $: 
    \begin{align*}
        \frac{D(U,V)}{\varepsilon} 
        &= \tr \left( \frac{U \oplus V - C}{\varepsilon} (\rho \otimes \sigma) \right) - \tr \left( \exp\left( \frac{U \oplus V - C}{\varepsilon} \right) \right) + 1
        \\
        &\phantom{=} + \frac{1}{\varepsilon} \tr \left( C (\rho \otimes \sigma) \right)
        \\
        &= \sum\limits_{j=1}^d \left( \frac{\kappa_j}{\varepsilon} \langle \xi_j | \rho \otimes \sigma | \xi_j \rangle - \exp\left(\frac{\kappa_j}{\varepsilon}\right) \right) + 1 + \frac{1}{\varepsilon} \tr \left( C (\rho \otimes \sigma) \right). 
    \end{align*}
    As $ \kappa_j \leq \lambda_{\max}( U \oplus V - C) $ and $ \exp\left(\frac{\kappa_j}{\varepsilon}\right) \geq \exp\left( \frac{\lambda_{\min} (U \oplus V - C)}{\varepsilon} \right) $ for all $ j$ as well as  $ \sum\limits_{j=1}^d  \left\langle \xi_j | \rho \otimes \sigma | \xi_j \right\rangle = \tr \left( \rho \otimes \sigma \right) = 1 $, one obtains 
    \begin{align*}
        \frac{D(U,V) - \tr \left( C (\rho \otimes \sigma) \right)}{\varepsilon} 
        &= 
        \sum\limits_{j=1}^d \left( \frac{\kappa_j}{\varepsilon} \langle \xi_j | \rho \otimes \sigma | \xi_j \rangle - \exp\left(\frac{\kappa_j}{\varepsilon}\right) \right) + 1
        \\
        &\leq 
        \frac{\lambda_{\max}( U \oplus V - C)}{\varepsilon} - \exp \left( \frac{\lambda_{\max}( U \oplus V - C)}{\varepsilon} \right) + 1
        \\
        &\phantom{\leq} - (d-1) \exp\left( \frac{\lambda_{\min} (U \oplus V - C)}{\varepsilon} \right),
    \end{align*}
    so 
    \begin{align*}
        &\phantom{\leq }\exp \left( \frac{\lambda_{\max}( U \oplus V - C)}{\varepsilon} \right) - \frac{\lambda_{\max}( U \oplus V - C)}{\varepsilon} - 1
        \\
        &\leq \frac{\tr \left( C (\rho \otimes \sigma) \right)- D(U,V) }{\varepsilon} 
        - (d-1) \exp\left( \frac{\lambda_{\min} (U \oplus V - C)}{\varepsilon} \right) 
    \end{align*}
    and by observing that $ \nu_2 $ is monotone, the claim follows. 
    
\end{proof}

\section{Qualitative proof for the convergence of Algorithm \ref{alg:bga}}
\begin{proof}
    First, we note that $ \eta_1$ and $ \eta_2$ are chosen so that by Corollary \ref{cor:dualImprovment}, one has 
    \begin{align*}
        D \left( U_{n+1}, V_{n+1} \right) \geq D \left(U_{n+1}, V_n\right) \geq D \left( U_n, V_n \right) 
    \end{align*}
    for all $ n$. 
    As $ D$ is bounded from above due to Theorem \ref{thm:Duality}, one has 
    \begin{align*}
    D \left(U_{n+1}, V_n\right) - D \left( U_n, V_n \right) \overset{n \to \infty}\longrightarrow 0
    \end{align*}
    and hence 
    \begin{align}
        \tr_2 \left( \exp \left( \frac{U_n \oplus V_n - C}{\varepsilon} \right) \right) \overset{n \to \infty}\longrightarrow \rho \label{eq:limTr2}
    \end{align}
    as $ E_1 (U_n,V_n) \overset{n \to \infty}\longrightarrow 0 $ by Corollary \ref{cor:dualImprovment}. Analogously, 
    \begin{align}
        \tr_1 \left( \exp \left( \frac{U_n \oplus V_n - C}{\varepsilon} \right) \right) \overset{n \to \infty}\longrightarrow \sigma, \label{eq:limTr1}
    \end{align}
    so that $ \nabla D (U_n, V_n) \overset{n \to \infty}\longrightarrow 0 $. 
    As $D$ is concave due to Lemma \ref{lem:DstronglyConcave}, this suffices to conclude 
    \begin{align}
        D(U_n, V_n) \overset{n \to \infty}\longrightarrow \max\limits_{U \in H(\mathcal{H}_1), V \in H(\mathcal{H}_2) } D(U,V). \label{eq:limUnVnMaximisesDual}
    \end{align}
     By Lemma \ref{lem:BoundEvs}, the set $ \left\{ U \oplus V \colon D(U,V) \geq M \right\} $ is compact for every $ M \in \R$. 
    Thus there are subsequences $ \left( U_{n_k} \right)_k \subset H(\mathcal{H}_1), \left( V_{n_k} \right)_k \subset H(\mathcal{H}_2)  $ and $ \hat U\in H(\mathcal{H}_1), \hat V  \in H(\mathcal{H}_2)$ such that 
    \begin{align*}
    U_{n_k} \oplus V_{n_k} \overset{k \to \infty}\longrightarrow \hat U \oplus \hat V. 
    \end{align*}
    As Equation (\ref{eq:limUnVnMaximisesDual}) holds also for every converging subsequence, by the uniqueness of the maximiser due to the strong concavity, one even has
    \begin{align*}
        U_n \oplus V_n \overset{n \to \infty}\longrightarrow \hat U \oplus \hat V
    \end{align*}
    and by continuity we get
    \begin{align*}
        D \left( \hat U, \hat V \right) = \max\limits_{U \in H(\mathcal{H}_1), V \in H(\mathcal{H}_2) } D(U,V)
    \end{align*}
    and for $ \Gamma := \exp \left( \frac{\hat U \oplus \hat V - C}{\varepsilon} \right) = \lim\limits_{n\to \infty} \exp \left( \frac{U_n \oplus V_n - C}{\varepsilon} \right)$ one also obtains
    $ F\left(\hat \Gamma \right) = \min\limits_{\Gamma \in  \mathfrak{P}(\mathcal{H}_1 \otimes \mathcal{H}_2), \Gamma \mapsto (\rho, \sigma)} F(\Gamma) $ by Theorem \ref{thm:Duality}. Finally, one has $ \Gamma \mapsto (\rho, \sigma)$ by Equation (\ref{eq:limTr2}) and Equation (\ref{eq:limTr1}) and by construction $ \Gamma \in \mathfrak{P} \left( \mathcal{H}_1 \otimes \mathcal{H}_2 \right) $ holds. 
\end{proof}

\end{document}